\documentclass{vldb}
\usepackage{graphicx}
\usepackage{balance}  
\usepackage{times,amsmath,epsfig,braket, amssymb, courier}
\usepackage{algorithm}
\usepackage{algpseudocode}
\usepackage{url}
\usepackage[figuresright]{rotating}
\usepackage{pgf, tikz}
\usepackage{multirow}

\usetikzlibrary{arrows, automata}
\usetikzlibrary{decorations.text}
\usetikzlibrary{positioning}

\let\svtikzpicture\tikzpicture
\def\tikzpicture{\noindent\svtikzpicture}
\algnewcommand\algorithmicforeach{\textbf{for each}}
\algdef{S}[FOR]{ForEach}[1]{\algorithmicforeach\ #1\ \algorithmicdo}

\newcommand{\quotes}[1]{``#1''}

\newtheorem{definition}{Definition}
\newtheorem{lemma}{Lemma}
\newtheorem{theorem}{Theorem}[section]
\begin{document}


\title{Graph-Based Proactive Secure Decomposition Algorithm for Context Dependent Attribute Based Inference Control Problem}



%
%
%
%

\numberofauthors{3} 

%
%
\author{
\alignauthor U\u{g}ur Turan\\
       \affaddr{Department of Computer Engineering}\\
       \affaddr{Middle East Technical University}\\
       \affaddr{06800, Ankara, Turkey}\\
       \email{ugur.turan@ceng.metu.edu.tr}
       \and
\alignauthor Ismail Hakk{\i} Toroslu \\
       \affaddr{Department of Computer Engineering}\\
       \affaddr{Middle East Technical University}\\
       \affaddr{06800, Ankara, Turkey}\\
       \email{toroslu@ceng.metu.edu.tr}
       \and
\alignauthor Murat Kantarc{\i}o\u{g}lu\\
       \affaddr{Department of Computer Science}\\
       \affaddr{University of Texas}\\
       \affaddr{at Dallas}\\
       \affaddr{Richardson, TX, USA}\\
       \email{muratk@utdallas.edu}
}

\maketitle

\begin{abstract}
Relational DBMS’s continue to dominate the database market, and inference problem on external schema of relational DBMS's is still an important issue in terms of data privacy. Especially for the last 10 years, external schema construction for application-specific database usage has increased its independency from the conceptual schema, as the definitions and implementations of views and procedures have been optimized. This paper offers an optimized decomposition strategy for the external schema, which concentrates on the privacy policy and required associations of attributes for the intended user roles. The method proposed in this article performs a proactive decomposition of the external schema, in order to satisfy both the forbidden and required associations of attributes. Functional dependency constraints of a database schema can be represented as a graph, in which vertices are attribute sets and edges are functional dependencies. In this representation, inference problem can be defined as a process of searching a subtree in the dependency graph containing the attributes that need to be related. The optimized decomposition process aims to generate an external schema, which guarantees the prevention of the inference of the forbidden attribute sets while guaranteeing the association of the required attribute sets with a minimal loss of possible association among other attributes, if the inhibited and required attribute sets are consistent with each other. Our technique is purely proactive, and can be viewed as a normalization process. Due to the usage independency of external schema construction tools, it can be easily applied to any existing systems without rewriting data access layer of applications. Our extensive experimental analysis shows the effectiveness of this optimized proactive strategy for a wide variety of logical schema volumes.
\end{abstract}

\section{Introduction}
As the demand towards automated systems and processes have increased, the technology in business applications has focused on to two different dimensions: application usage and statistical analysis. In each case, inference based privacy preserving techniques, in terms of databases, has been an important problem in order to protect the sensitive data. Modern approaches like differential privacy preserving techniques\cite{r1, rr13} or intentionally deception mechanisms provide secure ways to represent statistical results without revealing sensitive data. However their usage cannot be applied to traditional applications \cite{r21}. 

Many business applications aim to monitor and update single entity data. As an example, consider a call center module of a bank. When a customer wants to apply for a campaign, the operator should check her transaction history for prerequisites. Transactions are sensitive data, but, they cannot be altered by adding noise or hypothetical rows cannot be added for deception. The financial transactions should be viewed exactly as they are. If the task were a statistical analysis of transactions, then both techniques could have been applied to protect the inference of sensitive data. However these kinds of processes are mainly based on single entity business procedures, as standard application usage. That is, the call center employee should be able to access to a set of sensitive data according to assigned user role, and, the external layer of the database presented to this user role should not reveal any more information other than required. 

This objective needs three different perspectives. Firstly, the schema of external layer should be decomposed with a fine-grained attribute-based approach which preserves required associations for the user role and prevents any other inferences. This objective is the focus of this article, as necessary theorems and algorithms is proposed in this paper. The second perspective deals with inferences based on dynamic data distribution and the last one is about collaboration attacks \cite{r17,r31}. Both of them are also needed to be handled in order to satisfy privacy of sensitive data. The mechanisms for the last two inference channels can be applied as add-ons to the strategy given in this paper. However, the main and the first objective should be arranging the external layer for a specific user to prevent all unwanted inference operations. By definition, this is a proactive step, and it should be viewed as a policy-based normalization stage in terms of privacy. 

For this research, we have been motivated with a real life example. The problem was related with the recycling business application developed for a municipality in Antalya province in Turkey. The product was a web and point-of-sale application, in which, all citizens having smart cards were giving their recycling wastes to the waste-collector companies, and these companies were loading credits to the cards of citizens according to the current expected market value of the waste, also determined by type of the waste. In 2016, a citizen complained that she has been identified, and she is receiving messages in consistent with her consumption and recycling waste she has produced. It may be argued that many citizens may have nearly same consumption and thoroughly, waste statistics distributed within a year. However, the case is different as the claimant owns a hand-made gift company and has much more glass waste in November during preparation of gifts for new year’s day. Recycled waste collector companies were expected to query the time-based collection statistics of the trucks and the town management was also expected only to view the usage statistics of the system. Additionally, collector’s views are defined only as a subset of the town management external view set.

For a simplified description, the views are as follows:

\vspace{3mm}

[Available for the Collector and Town Management]

\vspace{1mm}

\texttt{View$_1$ = (\underline{TruckId, DateTime, GPSCoordinates},}

\hspace{8mm}\texttt{TotalWasteWeight, RecyclingWasteWeight,}

\hspace{8mm}\texttt{WasteType)}

\vspace{3mm}

[Available for only Town Management]

\vspace{1mm}

\texttt{View$_2$  =  (\underline{CitizenId}, Name, Surname, Address, }

\hspace{8mm}\texttt{PhoneNumber)}

\vspace{3mm}

A malicious worker in the town management can use GPSCoordinates attribute in association with the Address attribute and determine a small subset of citizens who had given glass waste in a period of time, using a simple GPS checking function $near$, such as: 

\vspace{3mm}

\texttt{Select * }

\texttt{From View$_1$, View$_2$ }

\texttt{Where near(GPSCoordinates, Address)}

\vspace{3mm}
Indeed, one malicious worker has shared this information with an advertisement company and they have used this information in favor of their glass-producer clients. It is not surprising to receive messages from other glass-producers for a company which always purchases glass. However this small gift company is a part of a charity organization and they use glasses they have collected from their members, and they do not purchase too much from the market. Therefore, it was not possible for glass-producers to determine this company for advertisement. As a result, the privacy of the citizen had been violated. The main reason behind the problem is the lack of security policy while defining views and the attribute association based cross-control in between views. As a solution, security dependent sets have been formed for the original database schema and the algorithms proposed in this paper have been applied. It should be noted that the relationship between GPSCoordinates and Address attributes is a probabilistic dependency, which can be treated as a kind of functional dependency. To overcome this privacy problem, following views are generated:

\vspace{3mm}

\texttt{New View$_1$  = (\underline{TruckId, DateTime}, WasteType, }

\hspace{8mm}\texttt{RecyclingWasteWeight, TotalWasteWeight)}

\vspace{3mm}

\texttt{New View$_2$ = (\underline{TruckId, GPSCoordinates})}

\vspace{3mm}

\texttt{New View$_3$  = (\underline{CitizenId}, Name, Surname, }

\hspace{8mm}\texttt{Address, PhoneNumber)}

\vspace{3mm}

Any join between “New View$_1$ ” and “New View$_2$ ” is not a meaningful join \footnote{Equi-join between primary and foreign keys} as TruckId is only foreign key, not a primary or candidate key, in these relations \cite{r12}.

This paper focuses on this problem and presents a complete mechanism to satisfy the goal. The core of the mechanism is based on the Functional Dependency Graph representation of database schema, which is constructed by defining attribute-sets as vertices and dependencies as edges. The aim is to find an external layer decomposition which strictly allows the required attribute associations and prevents inhibited associations, both in compliance with the privacy policy for a specific user role. Owing to the nature of domain, the proposed mechanism has attribute based granularity and the advantage of rearranging external layer without making any change on other layers. There may be other alternative approaches also for the kind of secure decomposition, introduced in this paper. The decomposition may be based on the forbidden attribute sets to satisfy the privacy or the required attributes sets needed for the user role \cite{r27}. The strategy given in this paper is an optimized combination of these two approaches and the most crucial step proposed in this paper, is to check the required sets together with the forbidden sets. This control step assures a fully compliant policy proactively.

The rest of this paper is organized as follows: the preliminaries and problem definition are given in the next section. In the following section, we first present graph-based representation of the problem and, then, formal definitions related to the problem and the algorithm for secure decomposition, which is based only on inhibited attribute sets, is given, which aims to minimize dependency loss. Afterward, the required attribute sets are defined, and the algorithm is extended to perform a complete policy-check and output a minimal secure decomposition, by the help of both user policy-based requirements and privacy policy. All algorithms are proven to produce secure decompositions. A brief related work section is given for literature review and experiments are also performed to show that the algorithm is also applicable even in large relational schemas. Future work and conclusions are given at the end of the paper.

\section{Preliminaries and Problem Definition}

Secure decomposition of the external schema to prevent unwanted inferences has been covered firstly in literature by \cite{A1} and \cite{A2}. These works concantrate on the required attribute sets (visibility constraints) and produce minimal sized fragments according to the dependencies and constraints. We have improved this process in our previous work \cite{r12} and we have defined, security dependency set concept and secure decomposition problem formally and we have proposed a decomposition algorithm. The algorithm aims to have maximal fragments (minimal dependency loss) according to forbidden set of attributes.

We use the following concepts (the original and fully formal definitions are in \cite{r12}) as:

\begin{itemize}
\item	\underline{Security dependent set} is the set of attributes from the logical schema, which should not be associated by using related schema and meaningful joins. The sets are determined logically with respect to the domain requirements.
\item	\underline{Meaningful join} is briefly an equi-join operation in between a foreign and primary keys. The inference problem is defined on inhibiting all possible meaningful joins, among the attributes of forbidden attribute sets.
\item	\underline{$R^+$} is given as the closure of all relations in a logical schema where the closure is determined by performing all possible meaningful joins to all relations of the schema.
\item	\underline{$F^+$} is the closure of all functional dependencies including the produced ones with transition, union and decomposition properties of functional dependencies.
\item	\underline{Identifiers} of an attribute are defined as the attribute set, to which the attribute is dependent to in $F^+$.
\item	\underline{Secure logical schema} is the one which prevents all associations of the security dependent sets by using $R^+$ and $F^+$. It is generated by a decomposition algorithm applied on the original schema. The algorithm is given in \cite{r12}.
\end{itemize}

The original definitions in \cite{r12} also includes the definition of probabilistic dependencies. For the sake of simplicity, they can be assumed to produce new security dependent sets as given in decomposition algorithm. Therefore, probabilistic dependencies are assumed to be inherently covered in all definitions and algorithms in this paper. Moreover, security dependent sets with a single attribute can easily be eliminated via removing this attribute from the schema, so security dependent sets are all assumed to be consisted of at least two attributes thereafter.

The decomposition algorithm, proposed in \cite{r12}, works as follows:

\begin{enumerate}
\item	Produce power set of all attributes for a relation.
\item	For each element set in this power set:
\begin{enumerate}
\item	Eliminate the element set, if it contains all attributes of any secure dependent set. 
\item	Eliminate the element set, if it contains any attribute in secure dependent set with its identifier.
\end{enumerate}
\item	Lastly, eliminate all element sets (trivial subsets), contained by other element sets.
\end{enumerate}

In order to illustrate the behavior of secure decomposition algorithm in \cite{r12}. We consider the following simple example:

Let the relation $R_k = \{A, B, C, D\}$ and $A$ is the primary key, single identifier for all other attributes, as the dependencies are illustrated in Figure~\ref{ex0}. Additionally, there is a single security dependent set as $\{B, C\}$.

\vspace{3mm}

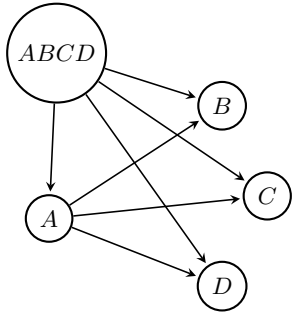
\begin{figure}
\centering
\begin{tikzpicture}[
            > = stealth, 
            shorten > = 1pt, 
            auto,
            semithick 
        ]

        \tikzstyle{every state}=[
            draw = black,
            thick,
            fill = white,
            minimum size = 4mm
        ]

			\node[state] (ABCD)  at (-0.4,-0.8) {$A B C D$};
        \node[state] (B) at (1.8,-1.5) {$B$};
        \node[state] (C) at (2.4,-2.7) {$C$};
			\node[state] (A) at (-0.5,-3) {$A$};
			\node[state] (D) at (1.8,-3.9) {$D$};

        \path[->] (ABCD) edge node {} (A);
			\path[->] (ABCD) edge node {} (B);
			\path[->] (ABCD) edge node {} (C);
			\path[->] (ABCD) edge node {} (D);
			\path[->] (A) edge node {} (B);
			\path[->] (A) edge node {} (C);
			\path[->] (A) edge node {} (D);

    \end{tikzpicture}
\caption{Dependency Graph}
\label{ex0}
\end{figure}

Therefore, the subsets containing both $B$ and $C$ needs to be eliminated in order to prevent the association between $B$ and $C$. Moreover, the subsets containing $\{A, B\}$ and $\{A, C\}$ will be eliminated as well, since $A$ is the identifier. After the elimination of trivial subsets, secure deposition of $R_k$ can be given as:

\vspace{3mm}

$R_{k_1} = \{A, D\}$	\hspace{3mm}  $R_{k_2} = \{B, D\}$	\hspace{3mm}  $R_{k_3} = \{C, D\}$

\vspace{3mm}

\begin{figure}
\centering
\begin{tikzpicture}[
            > = stealth, 
            shorten > = 1pt, 
            auto,
            semithick 
        ]

        \tikzstyle{every state}=[
            draw = black,
            thick,
            fill = white,
            minimum size = 4mm
        ]

			\node[state] (AD)  at (-2.4,-0.8) {$A D$};
        \node[state] (BD) at (-0.6,-0.8) {$B D$};
        \node[state] (CD) at (1.2,-0.8) {$C D$};
        \node[state] (B) at (0.6,-3) {$B$};
        \node[state] (C) at (2.4,-3) {$C$};
			\node[state] (A) at (-3,-3) {$A$};
			\node[state] (D) at (-1.2,-3) {$D$};

        \path[->] (AD) edge node {} (A);
			\path[->] (AD) edge node {} (D);
			\path[->] (A) edge node {} (D);
			\path[->] (BD) edge node {} (D);
			\path[->] (BD) edge node {} (B);
			\path[->] (CD) edge node {} (D);
			\path[->] (CD) edge node {} (C);

    \end{tikzpicture}
\caption{Dependency Graph After Strong-Cut}
\label{ex01}
\end{figure}
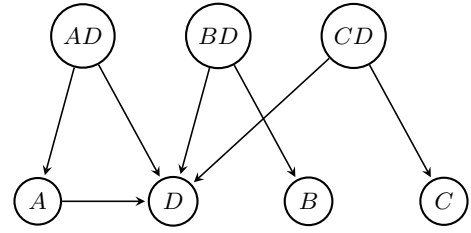
As it can be seen from the new decomposed schema, there is no way to perform a meaningful join between decomposed sets to associate $B$ and $C$. As a graph notation (details will be discussed later) Figure~\ref{ex01} shows the dependecnies formed after the decomposition and according to this figure, there is no way to associate $B$ and $C$ together, starting from a vertex in this graph. In other words, if the ways of associate secure dependent sets attributes are defined as a chain of functional dependencies through meaningful joins, the algorithm breaks these chains from both sides for both attributes. It is obvious that the relations containing the security dependent set should be removed, but the algorithm in \cite{r12}, breaks association of each attribute in security dependent set with its identifiers to prevent all meaningful joins. In this paper, we call this strategy as a $strong-cut$ approach.

However, this $strong-cut$ approach can be relaxed by cutting the chains only at a single point by producing:

\vspace{3mm}

$R_{k_1} = \{A, C, D\}$ 	\hspace{3mm}  $R_{k_2} = \{B, D\}$ 

\vspace{3mm}

The dependencies of this schema is depicted in Figure~\ref{ex02}. There is also another possible schema as:

\vspace{3mm}

$R_{k_1} = \{A, B, D\}$	\hspace{3mm}  $R_{k_2} = \{C, D\}$

\vspace{3mm}

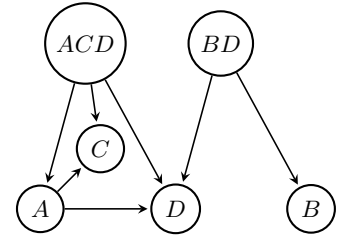
\begin{figure}
\centering
\begin{tikzpicture}[
            > = stealth, 
            shorten > = 1pt, 
            auto,
            semithick 
        ]

        \tikzstyle{every state}=[
            draw = black,
            thick,
            fill = white,
            minimum size = 4mm
        ]

			\node[state] (ACD)  at (-2.4,-0.8) {$A C D$};
        \node[state] (BD) at (-0.6,-0.8) {$B D$};
        \node[state] (B) at (0.6,-3) {$B$};
        \node[state] (C) at (-2.2,-2.2) {$C$};
			\node[state] (A) at (-3,-3) {$A$};
			\node[state] (D) at (-1.2,-3) {$D$};

        \path[->] (ACD) edge node {} (A);
			\path[->] (ACD) edge node {} (D);
			\path[->] (ACD) edge node {} (C);
			\path[->] (A) edge node {} (D);
			\path[->] (A) edge node {} (C);
			\path[->] (BD) edge node {} (D);
			\path[->] (BD) edge node {} (B);

    \end{tikzpicture}
\caption{Dependency Graph After Weak-Cut}
\label{ex02}
\end{figure}

\vspace{3mm}

Both of these schemas are consistent with the privacy constraint defined by security dependency set. The aim of this work is to develop a $relaxed-cut$ algorithm that decomposes the original schema with a minimum loss of functional dependencies while satisfying the security constraints.

The motivation of this work can be described as developing a decomposition that should not be much lossier than needed. This requirement defines an optimization problem and to the best of our knowledge, this is the first attempt in literature to construct an optimized secure decomposition satisfying the policy, while minimizing the dependency loss. 

Directed graph representation is selected as the most suitable mathematical model to represent the problem, since a functional dependency can be easily demonstrated as a directed edge and by this way, all algorithmic background in the graph theory can be used for further enhancements of the concept and algorithm. As a result, a new algorithm will be proposed for secure decomposition concept, which aims to decompose the original schema minimally by preserving the idea of prohibiting decomposed relations to be used in meaningful joins to associate the attributes in a security dependent set.

The problem is basically building a decomposition of the original schema, as any set of securely dependent attributes cannot be associated by joins on keys. 

This paper also introduces required attribte set definition and the relaxed-cut algorithm is improved with a consistency check among forbidden and required set attributes. 

\section{Relaxed-Cut Secure Decomposition Algorithm}

We firstly give the basic definitions by using graph notation.

\vspace{1.6mm}

\begin{definition}
\textbf{Functional Dependency Graph}
\textit{(denoted as FDG hereafter, for a schema)} The given functional dependency set (F) of a logical schema (where F is decomposed - i.e., there is a single element on the right hand side - and thus for each functional dependency $F_i$ such as  $A_i \rightarrow A_j$, $A_i$ is an attribute set and $A_j$ is a single attribute) can be represented as a directed graph $G = (V, E)$ as follows:

\begin{itemize}
\item	In a normalized schema, all attributes are expected to exist in $A_j$, but the schema may not be normalized, so each attribute should also be element of $V$ individually.
\item	Each one of $A_i$ and $A_j$ is a single node in $V$
\item	Each relation (attribute sets) is an individual node in $V$.
\item	Each dependencies of $F_i^+$ is an edge in $E$, if both sides of the dependency exist as a different node in $V$.
\end{itemize}
\end{definition}

\vspace{1.6 mm}

Example-1: Assume that the logical schema $S$ consists of four relations $R_1$, $R_2$, $R_3$ and $R_4$:

\vspace{3mm}

$R_1 = \{ A, B, C, D\}$  

$F_{R_1} = \{ A \rightarrow B, A \rightarrow C, A \rightarrow D,$ 
 
$ \quad \quad \quad \; \; D \rightarrow A, D \rightarrow B,  D \rightarrow C \}$

\vspace{3mm}

$R_2 = \{ E, F, G, H\}$  

$F_{R_2} = \{ E \rightarrow F, E \rightarrow G, E \rightarrow H,$

$ \quad \quad \quad \; \; H \rightarrow E, H \rightarrow F,  H \rightarrow G\}$

\vspace{3mm}

$R_3 = \{ A, E \}$  where $A$ and $E$ are foreign keys.

\vspace{3mm}

$R_4 = \{ H, D \}$  where $H$ and $D$ are foreign keys.

\vspace{3mm}

Then, the graph (FDG) constructed for this schema is given in Figure~\ref{ex1}.

\vspace{3mm}

\begin{figure}
\centering
\begin{tikzpicture}[
            > = stealth, 
            shorten > = 1pt, 
            auto,
            semithick 
        ]

        \tikzstyle{every state}=[
            draw = black,
            thick,
            fill = white,
            minimum size = 4mm
        ]

			\node[state] (EFGH)  at (-0.4,-0.8) {$E F G H$};
        \node[state] (F) at (1.8,-1.5) {$F$};
        \node[state] (G) at (2.4,-2.7) {$G$};
			\node[state] (E) at (-0.5,-3) {$E$};
			\node[state] (H) at (1.8,-3.9) {$H$};

        \path[->] (EFGH) edge node {} (E);
			\path[->] (EFGH) edge node {} (F);
			\path[->] (EFGH) edge node {} (G);
			\path[->] (EFGH) edge node {} (H);
			\path[->] (E) edge[red,very thick] node {} (F);
			\path[->] (E) edge node {} (G);
			\path[->] (E) edge node {} (H);
			\path[->] (H) edge node {} (F);
			\path[->] (H) edge node {} (G);
			\path[->] (H) edge node {} (E);

			\node[state] (ABCD)  at (4.6,-0.8) {$A B C D$};
        \node[state] (B) at (6.8,-1.5) {$B$};
        \node[state] (C) at (7.4,-2.7) {$C$};
			\node[state] (A) at (4.5,-3) {$A$};
			\node[state] (D) at (6.8,-3.9) {$D$};

        \path[->] (ABCD) edge node {} (A);
			\path[->] (ABCD) edge node {} (B);
			\path[->] (ABCD) edge node {} (C);
			\path[->] (ABCD) edge node {} (D);
			\path[->] (A) edge[red,very thick] node {} (B);
			\path[->] (A) edge node {} (C);
			\path[->] (A) edge node {} (D);
			\path[->] (D) edge node {} (B);
			\path[->] (D) edge node {} (C);
			\path[->] (D) edge node {} (A);

			\node[state] (AE) at (0.5,-6) {$A E$};

        \path[->] (AE) edge[red,very thick] node {} (A);

			\path[->] (AE) edge[red,very thick] node {} (E);

			\node[state] (HD) at (5,-6) {$H D$};
        \path[->] (HD) edge node {} (H);
			\path[->] (HD) edge node {} (D);

    \end{tikzpicture}
\caption{$FDG$ of Example-1}
\label{ex1}
\end{figure}
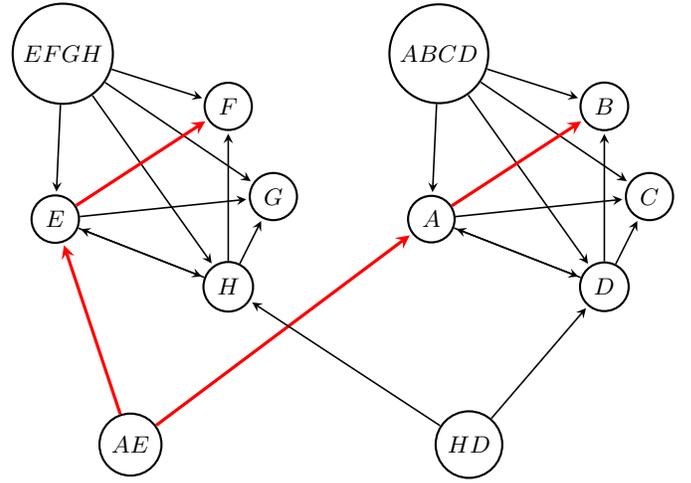

\begin{algorithm}[H]
\caption{Constructing Functional Dependency Graph (FDG)}
\label{array-sum}
\begin{algorithmic}[1]
\Require
 \Statex $S$: logical schema as $(R, F)$,
\Ensure
 \Statex $FGD_S = (V, E)$: functional dependency graph of $S$
    \State \textbf{begin}
    \State $V \leftarrow \{\}$ 
    \State $E \leftarrow \{\}$ 
    \State //Step-1
    \ForEach {$F_i (X_i \rightarrow Y_i) \in F$}
    		\If { $\left| Y_i \right| > 1$}
    			\State remove $F_i$ from $F$
    			\ForEach {$Y_j \in Y_i$}
    				\State add $X_i \rightarrow Y_j$ to $F$
    			\EndFor
    		\EndIf
    	\EndFor
    	\State //Step-2
    	\ForEach {$R_i \in R$}
    			\ForEach {$A_j \in R_i$}
    				\State add $A_j$ to $V$
    			\EndFor
    	\EndFor
    	\State //Step-3
	\ForEach {$F_i (X_i \rightarrow Y_i) \in F$}
    		\If { $\left| X_i \right| > 1$ and $X_i \notin V$}
    			\State add $X_i$ to $V$
    		\EndIf
    	\EndFor
    	\State //Step-4
	\ForEach {$R_i \in R$}
    		\If { $R_i \notin V$}
    			\State add $R_i$ to $V$
    		\EndIf
    	\EndFor
    	\State //Step-5
	\State $F^+ \leftarrow$ Closure set of $F$
	\State //Step-6
    	\ForEach {$F_i (X_i \rightarrow Y_i) \in F^+$}
    		\If { $X_i \neq Y_i$ and $X_i \in V$ and $Y_i \in V$}
    			\State add $X_i \rightarrow Y_i$ to $E$
    		\EndIf
    	\EndFor
	\State\textbf{end}
\end{algorithmic}
\end{algorithm}

The steps of the FDG construction algorithm is as follows:

\begin{enumerate}
\item	Decompose all functional dependencies in $F$, such that each functional dependency will have a single element in the right-hand side.
\item	Create an individual vertex for all attributes in schema and add to $V$.
\item	Create vertices for the attribute sets with more than one element, which exist in left-hand side of any functional dependency and does not exist in $V$.
\item	Create additional vertices, which include the attributes of a relation in $R$ and does not exist in $V$.
\item	Generate $F^+$
\item	For each $X \rightarrow Y$ in $F^+$, add an edge to $E$ if $X$ and $Y$ are different vertices in $V$.
\end{enumerate}

The graph in Figure~\ref{ex1} is obtained by the above algorithm for example-1.

\begin{lemma}

The edges of transitive closure of $FDG$ is equal to $F^+$

\end{lemma}

\begin{proof}

(SKETCH) It can be easily seen that the transitive definition is same for functional dependencies and its corresponding graph. The equivalency is based on the transitive property on the graphs and functional dependencies.

\end{proof}

\begin{definition}
\textbf{Common Ancestor of an Attribute Set}
 is a vertex in $FDG$, from which there exist simple paths to each element of attribute set.   
\end{definition}

In Figure~\ref{ex1}, $AE$ is one of the Common Ancestors for the set of vertices $\{ F, B \}$.

\begin{definition}
\textbf{Join Chain of an Attribute Set}
\textit{(Denoted as JC hereafter)} is the set of edges of simple paths in $FDG$, from a common ancestor to the attribute set.   

\end{definition}

The attribute sets may be a forbidden set (i.e. Secure Dependent Set) or a required set (which will be defined later). Let the relational schema is as given in Figure~\ref{ex1}. Let the forbidden set is $\{ F, B \}$, and the functional dependency graph is constructed as in Figure~\ref{ex1}. The join chain sets according to the forbidden set is given as below. The first join chain is emphasized with red colour and bold for an ilustrative example.

\vspace{3mm}

{\color{red} $JC_1 = \{ AE \rightarrow E, E \rightarrow F, AE \rightarrow A, A \rightarrow B \}$}

\vspace{3mm}

$JC_2 = \{ HD \rightarrow H, H \rightarrow F, HD \rightarrow D, D \rightarrow B \}$

\vspace{3mm}

$JC_3 = \{ AE \rightarrow E , E \rightarrow H, H \rightarrow F,$ 

$ \quad \quad \quad \; \; AE \rightarrow A, A \rightarrow D, D \rightarrow B \}$

\vspace{3mm}

$JC_4 = \{ AE \rightarrow E, E \rightarrow H, H \rightarrow F, AE \rightarrow A, A \rightarrow B \}$

\vspace{3mm}

$JC_5 = \{ AE \rightarrow E, E \rightarrow F, AE \rightarrow A, A \rightarrow D, D \rightarrow B \}$

\vspace{3mm}

$JC_6 = \{ HD \rightarrow H, H \rightarrow E, E \rightarrow F,$

$ \quad \quad \quad \; \;  HD \rightarrow D, D \rightarrow A, A \rightarrow B \}$

\vspace{3mm}

$JC_7 = \{ HD \rightarrow H, H \rightarrow F, HD \rightarrow D, D \rightarrow A, A \rightarrow B \}$

\vspace{3mm}

$JC_8 = \{ HD \rightarrow H, H \rightarrow E, E \rightarrow F, HD \rightarrow D, D \rightarrow B \}$

\vspace{3mm}

\begin{algorithm}[H]
\caption{Generating Join Chain Set Algorithm}
\label{array-sum}
\begin{algorithmic}[1]
\Require
 \Statex $FDG$: functional dependency graph$( V, E)$ of schema
 \Statex $A$: attribute set
\Ensure
 \Statex $JC$: join chain set
 	\State //Step-1
    \State \textbf{begin}
        \State $JC \leftarrow \{\}$ 
    \ForEach {$(X_i \rightarrow Y_i) \in E$}
    			\State remove $X_i \rightarrow Y_i$ from $E$
    			\State add $Y_i \rightarrow X_i$ to $E$
    	\EndFor
    	\State //Step-2
	\State Initialize $TargetArr$ as array of array of vertices
	\State Initialize $PathArr$ as array of array of set of edges
	\ForEach {$A_i \in A$}
		\State $C_{A_i} \leftarrow$ empty set of connected vertices 
		\State $P_{A_i} \leftarrow$ empty set of path edge sets 
		\State $C_{A_i}, P_{A_i} \leftarrow$ apply $DFS$ to $FDG$ with starting vertex $A_i$
		\State $j \leftarrow 0$
    		\ForEach {$C_j$ in $C_{A_i}$}
    			\State $TargetArr[i][j] \leftarrow C_j$
   			\State $PathArr[i][j] \leftarrow$ simple path to $C_j$ in $P_{A_i}$
   			\State $j \leftarrow j + 1$
    		\EndFor
    	\EndFor
    	\State //Step-3
	\State $SharedArr \leftarrow$ array of shared vertices by all TargetArr rows		
	\ForEach {$S_i \in SharedArr$}
    			\State add $\bigcup(PathArr[k][l] \text{ as }  TargetArr[k][l] = S_i)$ to $JC$
    	\EndFor
    	\State //Step-4
	\ForEach {$JC_i \in JC$}
		\ForEach {$JC_j \in JC$}
			\If { $JC_i \supseteq JC_j$}
    				\State remove $JC_i$ from $JC$
    			\EndIf
		\EndFor
	\EndFor
	\State\textbf{end}
\end{algorithmic}
\end{algorithm}

The steps of the Join Chain Construction algorithm is as follows:

\begin{enumerate}
\item	All edges are reversed.
\item	Taking each element of attribute set ($A$) as starting vertex, apply $DFS$ up to all connected vertices and all possible simple paths are determined for each end vertex.
\item	If there exist simple paths to the same end vertex, which are common (common ancestors in original $FDG$) for all set attributes (assumed to be starting vertices), all combinations of constructed simple paths, starting from different set attribute and ending in the same vertex is a join chain.
\item	If a chain composes another chain, it is discarded.
\end{enumerate}

In example-1, $HD$ and $AE$ are determined as common ancestors and all combinations of simple paths as $AE \rightarrow B$ (2 alternatives) \& $AE \rightarrow F$ (2 alternatives) and $HD \rightarrow B$ (2 alternatives) and $HD \rightarrow F$ (2 alternatives) are given as different join chains.

\begin{definition}
\textbf{Minimum-Cut Secure Decomposition:} Decomposing the relational schema by removing the minimum number of functional dependencies (i.e., not allowing the attributes of the lost functional dependency $A_i \rightarrow A_j$ in the same relation) to satisfy all security requirements.
\end{definition}

Minimized-Cut Secure Decomposition Problem is equivalent to Minimum Hitting Set Problem \cite{minhit} and thus it is NP-Complete. We propose a simple greedy heuristic algorithm to solve Relaxed-Cut Secure Decomposition problem (which is defined below), and due to the structure of our problem, we observed that this greedy approach mostly determines the optimal solution.

\begin{definition}
\textbf{Relaxed-Cut Secure Decomposition:} Decomposing the relational schema by greedly removing the functional dependencies in ordet to cut the dependencies of secure dependent attributes at least through one of the join chains.     
\end{definition}

Unlike strong cut which cuts the identifiers of all attributes of the secure dependent sets, relaxed cut aims to remove the functional dependencies as little as possible. The steps of the Relaxed-Cut Secure Decomposition algorithm are as follows:

\begin{enumerate}
\item	Calculate all join chains( $JC_i$ ) for each security dependent set (Algorithm-2).
\item	For each edge in the $FDG$, determine the number of times (SecurityCount) it appears in join chains.
\item	Sort the edges first according to their SecurityCount in descending order, then, the number of attributes on the nodes at both sides of the edge, in ascending order (in order to cut lower chains first).
\item	Traverse the sorted list and mark each join chain as cut, if the edge is contained. These edges are selected ones and to be a selected one, an edge should be an element of at least one unmarked join chain. Set of selected edges are named as new security dependent sets.
\item	All subsets of the attributes of the relational schema are generated, which is called as $PSR_x$ in the algorithm. Then, for each new security dependency set, each element of $PSR_x$ is processed. The element set is eliminated if it contains all attributes of that security dependency set together.
\item	After that, among the remaining subsets redundant ones (used for unnecessary sub-relations composed by other sub-relations) are also eliminated.
\end{enumerate}

The steps 1 through 4 can be named as “Relaxation Stage” and steps 5 and 6 as “Decomposition Stage”. Decomposition stage is a subpart of the secure decomposition algorithm proposed in \cite{r12} except the identifier elimination stage. 

\vspace{3mm}

\begin{algorithm}[H]
\caption{Relaxed-Cut Secure Decomposition Algorithm}
\label{array-sum}
\begin{algorithmic}[1]
\Require
 \Statex $LS$: logical schema as $( R, F )$
 \Statex $SD$: set of security dependent sets for $LS$
\Ensure
 \Statex $PS_R$: a subset of maximal subsets of $R$ satisfying security decomposition constraints

    \State \textbf{begin}
     \State $JC \leftarrow$ empty array of join chain sets 
     \State //Step-1
    \ForEach {$SD_i \in SD$}
    			\State $JC_i \leftarrow$  join chain set for $SD_i$ (Algorithm-2)
    			\State add $JC_i$ to $JC$
    	\EndFor
    	\State $FDG ( V, E) \leftarrow$ FDG of $LS$ (Algorithm-1) 
	\State $SecurityCount \leftarrow$ array of integers initialized to 0, size $\left|E\right|$
	\State //Step-2
	\ForEach {$E_i \in E$}
		\ForEach {$JC_k \in JC$}
			\ForEach {$JC_{k_i} \in JC_k$}
				\If { $E_i \in JC_{k_i}$}
    					\State $SecurityCount[i]++$
    				\EndIf
    			\EndFor
    		\EndFor
    	\EndFor
	\State //Step-3
    	\State $E \leftarrow$ sort edges in descending order (uses $SecurityCount$)
	\State $Selection \leftarrow$ empty set of edges
	\ForEach {$E_i \in E$} //Step-4
		\ForEach {$JC_k \in JC$}
			\ForEach {$JC_{k_i} \in JC_k$}
				\If { $E_i \in JC_{k_i}$ and $JC_{k_i}$ is unmarked}
    					\State mark $JC_{k_i}$
    					\State add $E$ to $Selection$
    				\EndIf
    			\EndFor
    		\EndFor
    	\EndFor
	\State $SD_{new} \leftarrow$ empty set of attribute sets
	\ForEach {$(X_i \rightarrow Y_i) \in Selection$}
    			\State add $(X_i \cup Y_i)$ to $SD_{new}$
    	\EndFor
	\ForEach {$R_x \in R$} //Step-5
		\State $PSR_x \leftarrow$ Power Set of $R_x$
		\ForEach {$D_i \in SD_new$}
			\ForEach {$S_j \in PSR_x$}
				\If { $D_i \subseteq S_j$ }
    					\State remove $S_j$ from $PSR_x$
    				\EndIf
    			\EndFor
    		\EndFor
    		\ForEach {$SS_i \in PSR_x$} //Step-6
			\ForEach {$S_j \in SS_i$}
				\If { $S_j \subseteq SS_i$ }
    					\State remove $S_j$ from $PSR_x$
    				\EndIf
    			\EndFor
    		\EndFor
    		
    	\EndFor
	\State\textbf{end}
\end{algorithmic}
\end{algorithm}

\begin{figure*}[!tb]
\centering
\begin{tikzpicture}[
            > = stealth, 
            shorten > = 1pt, 
            auto,
            node distance = 2.5cm, 
            semithick 
        ]

        \tikzstyle{every state}=[
            draw = black,
            thick,
            fill = white,
            minimum size = 4mm
        ]

			\node[state] (MHRJ)  at (-0.4,-0.8) {$M H R J$};
        \node[state] (H) at (-1.6,-4.8){$H$};
        \node[state] (R) at (0.6,-4.8) {$R$};
			\node[state] (J) at (-2.4,-2.4) {$J$};
			\node[state] (M) at (2.4,-2.8) {$M$};

        \path[->] (MHRJ) edge node [pos=0.3, fill=white, scale=0.9, anchor=center]{$e_{22}$} (J);
			\path[->] (MHRJ) edge node [pos=0.2, fill=white, scale=0.9, anchor=center]{$e_{19}$} (H);
			\path[->] (MHRJ) edge node [pos=0.2, fill=white, scale=0.9, anchor=center]{$e_{21}$} (R);
			\path[->] (MHRJ) edge node [pos=0.3, fill=white, scale=0.9, anchor=center]{$e_{20}$} (M);
			\path[->] (M) edge node [pos=0.3, fill=white, scale=0.9, anchor=center]{$e_{23}$} (H);
			\path[->] (M) edge node [pos=0.3, fill=white, scale=0.9, anchor=center]{$e_{24}$} (R);
			\path[->] (M) edge node [pos=0.3, fill=white, scale=0.9, anchor=center]{$e_{25}$} (J);

			\node[state] (K) at (-4.4,-3.4) {$K$};
			\node[state] (L) at (-3.0,-3.8) {$L$};
			\node[state] (JKL)  at (-4.4,-0.8) {$J K L$};

			\path[->] (JKL) edge node [pos=0.3, fill=white, scale=0.9, anchor=center]{$e_{26}$} (J);
			\path[->] (JKL) edge node [pos=0.3, fill=white, scale=0.9, anchor=center]{$e_{27}$} (L);
			\path[->] (JKL) edge node [pos=0.3, fill=white, scale=0.9, anchor=center]{$e_{28}$} (K);
			\path[->] (J) edge node [pos=0.3, fill=white, scale=0.9, anchor=center]{$e_{29}$} (K);
			\path[->] (J) edge node [pos=0.3, fill=white, scale=0.9, anchor=center]{$e_{30}$} (L);

			\node[state] (ABCD)  at (7.6,-0.8) {$A B C D$};
        \node[state] (B) at (9.4,-2.0) {$B$};
        \node[state] (C) at (10,-4) {$C$};
			\node[state] (A) at (7.1,-4) {$A$};
			\node[state] (D) at (8,-6) {$D$};

        \path[->] (ABCD) edge node [pos=0.3, fill=white, scale=0.9, anchor=center]{$e_4$} (A);
			\path[->] (ABCD) edge node [pos=0.4, fill=white, scale=0.9, anchor=center]{$e_6$} (B);
			\path[->] (ABCD) edge node [pos=0.2, fill=white, scale=0.9, anchor=center]{$e_7$} (C);
			\path[->] (ABCD) edge node [pos=0.4, fill=white, scale=0.9, anchor=center]{$e_5$} (D);
			\path[->] (A) edge[bend right] node [pos=0.3, fill=white, scale=0.9, anchor=center]{$e_{10}$} (B);
			\path[->] (A) edge node [pos=0.4, fill=white, scale=0.9, anchor=center]{$e_{11}$} (C);
			\path[->] (A) edge node [pos=0.3, fill=white, scale=0.9, anchor=center]{$e_8$} (D);
			\path[->] (B) edge node [pos=0.8, fill=white, scale=0.9, anchor=center]{$e_{12}$} (D);
			\path[->] (B) edge node [pos=0.3, fill=white, scale=0.9, anchor=center]{$e_{13}$} (C);
			\path[->] (B) edge[bend right] node [pos=0.4, fill=white, scale=0.9, anchor=center]{$e_9$} (A);

			\node[state] (AEM)  at (3.6,-0.8) {$A E M$};
			\node[state] (EFG)  at (4.2,-5.8) {$E F G$};
        \node[state] (F) at (3.3,-4.2){$F$};
        \node[state] (G) at (5.3,-4.2) {$G$};
			\node[state] (E) at (4.3,-2.6) {$E$};

        \path[->] (EFG) edge node [pos=0.3, fill=white, scale=0.9, anchor=center]{$e_{17}$} (E);
        \path[->] (EFG) edge node [pos=0.4, fill=white, scale=0.9, anchor=center]{$e_{16}$} (F);
        \path[->] (EFG) edge node [pos=0.4, fill=white, scale=0.9, anchor=center]{$e_{18}$} (G);
        \path[->] (AEM) edge node [pos=0.4, fill=white, scale=0.9, anchor=center]{$e_3$} (M);
        \path[->] (AEM) edge node [pos=0.4, fill=white, scale=0.9, anchor=center]{$e_1$} (A);
        \path[->] (AEM) edge node [pos=0.4, fill=white, scale=0.9, anchor=center]{$e_2$} (E);
			\path[->] (E) edge node [pos=0.5, fill=white, scale=0.9, anchor=center]{$e_{14}$} (F);
			\path[->] (E) edge node [pos=0.5, fill=white, scale=0.9, anchor=center]{$e_{15}$} (G);

    \end{tikzpicture}
    \caption{$FDG$ of Example-3}
\label{ex3}
\end{figure*}
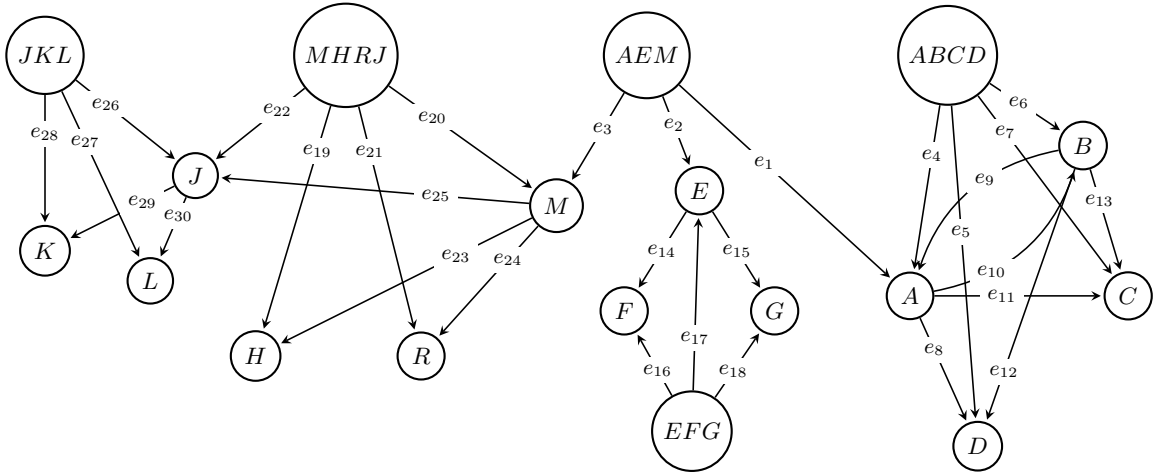

\vspace{3mm}

Example-2: Consider the following schema and the forbidden sets.

\vspace{3mm}

$R_1 = \{ \underline{A}, \underline{B}, C, D \}$ 

\vspace{3mm}

$R_2 = \{ \underline{E}, F, G \}$

\vspace{3mm}

$R_3 = \{ \underline{A, E, M} \}$

\vspace{3mm}

$R_4 = \{ \underline{J}, K, L \}$

\vspace{3mm}

$R_5 = \{ \underline{M}, H, R, J \}$

\vspace{3mm}

Forbidden Sets =  $\{ A, D \}, \{ D, F \}, \{ K, H \}$

\vspace{3mm}

Functional dependency graph is constructed as in Figure~\ref{ex3}.

Join chains are determine for this example as follows:

\vspace{3mm}

For $\{ A, D \}$:

\vspace{3mm}

\hspace*{3mm} $JC_1 = \{ e_8 \}$ \hspace{3mm} $JC_2 = \{ e_{10}, e_{12} \}$ 

\vspace{3mm}

\hspace*{3mm} $JC_3 = \{ e_4, e_5 \}$ \hspace{3mm} $JC_4 = \{ e_6, e_{12}, e_4 \}$ 

\vspace{3mm}

\hspace*{3mm} $JC_5 = \{ e_5, e_6, e_9 \}$ \hspace{3mm} $JC_6 = \{ e_9, e_{12} \}$

\vspace{3mm}

For $\{ D, F \}$:

\vspace{3mm}

\hspace*{3mm} $JC_7 = \{ e_1, e_2, e_8, e_{14} \}$ \hspace{3mm} $JC_8 = \{ e_1, e_2, e_{10}, e_{12}, e_{14} \}$

\vspace{3mm}

For $\{ K, H \}$:

\vspace{3mm}

\hspace*{3mm} $JC_9 = \{ e_{23}, e_{25}, e_{29} \}$ \hspace{3mm} $JC_{10} = \{ e_{19}, e_{22}, e_{29} \}$ 

\vspace{3mm}

\hspace*{3mm} $JC_{11} = \{ e_{20}, e_{22}, e_{23}, e_{29} \}$


\vspace{3mm}

Finally, relaxed cut secure dependency algorithm is executed using Security Counts. The edges are shown up to all marked ones in Table-1. Plus (+) sign in a row indicates that this edge is selected and the join chain on the column is marked. Minus (-) sign is used for already marked join chains and the edges without plus (+) sign in the row is not selected.

As a result, the following edges are selected: 

\vspace{2mm}

$\{ B \rightarrow D, A \rightarrow D, J \rightarrow K, B \rightarrow A, ABCD \rightarrow A \}$

\vspace{3mm}

Output of Algorithm in \cite{r12} with the Security Dependent Set $\{ A, D \}, \{ D, F \}, \{ K, H \}$ would be as (i.e. with strong-cut): 

\vspace{3mm}

\hspace*{3mm} $R_{1_1} = \{ A, C \}$ \hspace{3mm} $R_{1_2} = \{ B, C \}$ \hspace{3mm} $R_{1_3} = \{ C, D \}$ 

\vspace{3mm}

\hspace*{3mm} $R_{2_1} = \{ E, G \}$ \hspace{3mm} $R_{2_2} = \{ F, G \}$

\vspace{3mm}

\hspace*{3mm} $R_3 = \{ A, E, M \} $

\vspace{3mm}

\hspace*{3mm} $R_{4_1} = \{ J, L \}$ \hspace{4.4mm} $R_{4_2} = \{ K, L \}$

\vspace{3mm}

\hspace*{3mm} $R_{5_1} = \{ M, R, J \}$  $R_{5_2} = \{ J, R, H \}$

\vspace{3mm}

On the other hand the output of Algorithm-3 with Security Dependent Sets 
$\{ B, D \}, \{ A, D \}, \{ J, K \}, \{ A, B, C, D \}$ will be as follows:

\vspace{3mm}

\hspace*{3mm} $R_{11} = \{ A, C \}$ \hspace{3mm} $R_{12} = \{ B, C \}$ \hspace{3mm} $R_{13} = \{ C, D \}$ 

\vspace{3mm}

\hspace*{3mm} $R_{2} = \{ E, F, G \}$ 

\vspace{3mm}

\hspace*{3mm} $R_3 = \{ A, E, M \} $

\vspace{3mm}

\hspace*{3mm} $R_{41} = \{ J, L \}$ \hspace{3mm} $R_{42} = \{ K, L \}$ 

\vspace{3mm}

\hspace*{3mm} $R_{5} = \{ M, H, R, J \}$

\vspace{3mm}

The algorithm can be improved by defining a total participation count to all edges for all possible join chains of combination of attributes but it will result in a high time-cost.  


\begin{theorem}
Algorithm-3 generates a secure logical schema.
\end{theorem}

\begin{proof}
Assume that the resulting decomposed relations can be joined by foreign keys to associate securely dependent attributes. Then the functional dependency graph of new schema should contain a join chain for the attributes of this security dependent set. However, this cannot happen since each join chain has been cut at least by an edge and the attributes of these cut edges are given as new security dependent sets, which means their coexistence is prevented.

Therefore, resulting decomposed relations form a secure logical schema and the new forbidden sets, serve for the same privacy degree with respect to secure decomposition in \cite{r12} by using original forbidden sets.

\end{proof}

\section{Secure Decomposition with Required Attribute Sets}

Considering the frontend applicational usage of relational databases, all roles and permitted functionalities are predetermined on the requirement analysis and design stages of the project. Thus, each functionality may be represented as a set of database queries and each query can be shown as an attribute set, which should be associated to accomplish the task. However, these required sets should be checked against the security dependent sets as a verification step. If any inconsistency is determined, the design of queries or security policy should be reviewed. In addition to these, the decomposition alternatives (according to the edge selection strategy in Algorithm-3) can be quantified and chosen according to the given association sets, which results in a decomposition which satisfies both security dependencies and needed functionalities. 

First, we define required attribute set:

\begin{definition}
\textbf{Required Attribute Set}
\textit{(Denoted as $RS$ hereafter)} is a set of attributes in the relational schema, which should be associated with a series of meaningful joins to satisfy a functionality of the applicational usage.   

\end{definition}

It is important to note that, each functionality of a user role should be mapped to a set of $RS$.

Required and forbidden sets must be consistent, and the decompostion algorithm must satisfy both requirements.

\begin{definition}
\textbf{Consistency Check Between Required Sets and Forbidden Sets:}
Required and Forbidden Sets are consistent with each other if there is a "cut set" that contains at least one element from each one of the join chains for each forbidden set (each forbidden set forms a set of join chains) and there is at least one set in join chains corresponding to each required set (each required set forms a set of join chains) that do not contain any element from the cut set.
\end{definition}

In the case of any inconsistency discovered, security policy and association set needs should be revised by the designer.

Consistency check (CC) problem can be simplified as follows: the edges are mapped to letters, and thus join chains become set of letters. The consistency check problem can be defined as to determine a set of letters (cut set) such that at least one letter from each set of the forbidden sets and none of the letters in at least one of the sets of each required set must be in the cut set.

Consider the following CC problem instance:

\vspace{4mm}

$\text{Join Chains}_{FS} = \left\lbrace 
\begin{aligned}
\{e, f, g\}  \hspace{0.1cm} \\
\{a, b, c, d\}  \\
\{a, e, c, d\}   \\
\{b, f, g\} \hspace{0.1cm}
\end{aligned} \right\rbrace$

\vspace{3mm}

$\text{Join Chains}_{RS} = \left\lbrace 
\begin{aligned}
	\left\lbrace 
	\begin{aligned}
		\{a, b, f, g\}  \\
		\{b, c, f\} \hspace{0.2cm}
	\end{aligned} 
	\right\rbrace \\
\left\lbrace 
	\begin{aligned}
		\{d, e\}  \\
		\{c, g\}
	\end{aligned} 
	\right\rbrace \hspace{0.3cm}
\end{aligned} \right\rbrace$

\vspace{3mm}

$\mathcal{\text{Cut-Set}} = \left\lbrace a, g \right\rbrace$ 

$\mathcal{\text{Preserved Required Words}} = \left\lbrace \{d, e\}, \{b, c, f\} \right\rbrace$

\vspace{2mm}

The above CC instance is consistents since at least one element of the cut set \{a, g\}  is in each forbidden set, and one set for each set of required sets (as \{d, e\} and \{b, c, f\}) do not contain any element of the cut set.

On the other hand, the following CC instance is inconsistent since no cut set satisfying the requirements exists.

\vspace{4mm}

$\text{Join Chains}_{FS} = \left\lbrace 
\begin{aligned}
\{e, f, g\}  \hspace{0.1cm} \\
\{a, b, c, d\}  \\
\{a, e, c, d\}   \\
\{b, f, g\} \hspace{0.1cm}
\end{aligned} \right\rbrace$ 

\vspace{3mm}

$\text{Join Chains}_{RS} = \left\lbrace 
\begin{aligned}
	\left\lbrace 
	\begin{aligned}
		\{a, b, f, g\}  \\
		\{b, f\} \hspace{0.3cm}
	\end{aligned} 
	\right\rbrace \\
\left\lbrace 
	\begin{aligned}
		\{d, g\}  \\
		\{c, g\}
	\end{aligned} 
	\right\rbrace \hspace{0.3cm}
\end{aligned} \right\rbrace$ 

\vspace{3mm}

$\mathcal{\text{Found}} = \text{Inconsistency}$ \hspace{0.5cm}

\vspace{3mm}
\begin{theorem}
Consistency Check ($CC$) Problem is NP-Complete.
\end{theorem}

\begin{proof}

Input: Set of sets $A$ and set of sets of sets $B$ as follows:

\vspace{3mm}

$\mathcal{A} = \left\lbrace \left\lbrace a_{1_1}, a_{1_2}, ...\right\rbrace, \left\lbrace a_{2_1}, ...\right\rbrace, ...\right\rbrace$   

\vspace{2mm}

$\mathcal{B} = \left\lbrace \left\lbrace \left\lbrace b_{1_{1_1}}, b_{1_{1_2}}, ...\right\rbrace, \left\lbrace b_{1_{2_1}}, ...\right\rbrace, ...\right\rbrace, \left\lbrace \left\lbrace b_{2_{1_1}}, ...\right\rbrace, ...\right\rbrace, ...\right\rbrace$

\vspace{2mm}

Problem: Given $A$ and $B$ as above determine (i.e. it is consistent), if there is a set of $a_{i_j}$ for each $i$ in $A$, there is at least one $r$ for each $k$ in $B$ such that $b_{k_r}$ does not include any of $a_{i_j}$.

\vspace{2mm}

More formally;

\vspace{1mm}

Given $(A,B)$ sets, the system is consistent

\vspace{1mm}

if there exists $C = \{ a_{i_{j}} | \forall i \exists j (a_{i_{j}} \in A) \}$ 

\vspace{1mm}

such that $\forall k \exists r (b_{k_r} \in B \wedge b_{k_r} \cap C = \emptyset)$

\vspace{2mm}

NPC Proof: Given 3SAT instance constuct an instance of $CC$ as follows:

\vspace{2mm}

3SAT instance: $INS_{3SAT}=(p_{1_1} \lor p_{1_2} \lor p_{1_3}) \land (p_{2_1} \lor p_{2_2} \lor p_{2_3}) \land ... (p_{s_1} ..)$ such that each $p_{i_j}$ is either $q_k$ or $\neg q_k$ and there are exactly $k$ different propositional variables {$q_1$,...$q_k$}

\vspace{2mm}

Construct instance $INS_{CC}$ as follows:

\vspace{2mm}

$\mathcal{\text{Generate A}} = \left\lbrace \left\lbrace q_1, \neg q_1 \right\rbrace, ..., \left\lbrace q_k, \neg q_k \right\rbrace \right\rbrace$

\vspace{2mm}

$\mathcal{\text{Generate B}} = \left\lbrace \left\lbrace \left\lbrace p_{1_1} \right\rbrace, \left\lbrace p_{1_2} \right\rbrace, \left\lbrace p_{1_3} \right\rbrace \right\rbrace, ...\right\rbrace$

\vspace{3mm}

$INS_{3SAT}$ is satisfiable if and only if $INS_{CC}$ is consistent.

$INS_{3SAT}$ is satisfiable if there is a truth assignment for each literal to make all clauses as true. This is equivalent to $INS_{CC}$, such that for each literal one element from each $\{q_i,  \neg q_i \}$ is selected to be in set $C$ (that is equivalent to false in $INS_{3SAT}$ or an edge to be cut in original consistency check problem). If each clause in satisfiable in $INS_{3SAT}$, then, at least 1 literal of $p_{i_1}$ or $p_{i_2}$ or $p_{i_3}$ must be true. That means either $\{p_{i_1}\} \cap C = \emptyset$ or $\{p_{i_2}\} \cap C = \emptyset$ or $\{p_{i_3}\} \cap C = \emptyset$.

\end{proof}

The following algorithm checks for inconsistency for given required sets $RS$ and forbidden sets $SD$. The initial iterations are the same to find a suitable decomposition, if exists. 

The steps of the Algorithm-4 is as follows:

\begin{enumerate}
\item	Calculate all join chains($JC_x$) for each security dependent set $SD$ (Algorithm-2). 
\item	Calculate all join chains($JC_y$) for each required attribute set $RS$ (Algorithm-2).
\item	For each different edge combination set, in which each single edge is chosen from a single join chain in all $JC_y$’s, check if these selected edges are unbroken while breaking all forbidden sets's join chains, then there is a disjoint edge set which breaks all join chains in $JC_x$. If such an edge combination cannot be found then inconsistency exists.
\end{enumerate}

The process will continue with decomposition stage of Algorithm-3 with $CutSet$ output to find a secure decomposition.

\begin{algorithm}[H]
\caption{Consistency Check and Determining Cut Set Algorithm}
\label{array-sum}
\begin{algorithmic}[1]
\Require
 \Statex $RS$: set of required sets for logical schema 
 \Statex $SD$: set of security dependent sets for logical schema
\Ensure
 \Statex $Inconsistent$: true or false
 \Statex $CutSet$: possible forbidden sets for decomposition
    \State \textbf{begin}
    \State $JC_x \leftarrow$ set of join chain sets
	\State $JC_y \leftarrow$ array of set of join chain sets
	\ForEach {$SD_i \in SD$} //Step-1
    			\State $Frb_i \leftarrow$  join chain sets for $SD_i$ (Algorithm-2)
    			\State add $Frb_i$ to $JC_x$
    	\EndFor
    	\State $j \leftarrow 0$
	\ForEach {$RS_i \in RS$} //Step-2
    			\State $Req_i \leftarrow$  join chain sets for $AS_i$ (Algorithm-2)
    			\State add $Req_i$ to $JC_{y[j]}$
    			\State $j \leftarrow j + 1$
    	\EndFor
	\ForEach {possible edge set $CutSet$, as one edge is selected from an 
	\State element of all $JC_y$ members} //Step-3
    			\State $found \leftarrow$  true
    			\ForEach {$JC_{x_k} \in JC_x$}
    				\If { $JC_{x_k} \cap CutSet \neq \emptyset$}
    					\State $found \leftarrow$  false
    					\State break
    				\EndIf
    			\EndFor
    			\If { $found$ }
    					\State break
    				\EndIf
    	\EndFor
    
	\State\textbf{end}
\end{algorithmic}
\end{algorithm}

\vspace{3mm}

{\renewcommand{\arraystretch}{2}%
\begin{table*}
\centering
    \caption{Greedy Edge Selection Phase} \begin{small}
    \begin{tabular}{|c|c|c|c|c|c|c|c|c|c|c|c|c|c|}
    \hline
     {\bfseries EDGE } & {\bfseries EDGE NUMBER } & {\bfseries SecurityCount} & {\bfseries $JC_1$} & {\bfseries $JC_2$} 
     & {\bfseries $JC_3$ } & {\bfseries $JC_4$} & {\bfseries $JC_5$} & {\bfseries $JC_6$} 
     & {\bfseries $JC_7$ } & {\bfseries $JC_8$} & {\bfseries $JC_9$} & {\bfseries $JC_{10}$}
     & {\bfseries $JC_{11}$ }
     \\
     
    \hline
    \textbf{$B \rightarrow D$}  & \textbf{$e_{12}$}  & \textbf{4} &  & +  &  & + &  & + & & + & & & \\
    \hline
    \textbf{$J \rightarrow K$}  & \textbf{$e_{29}$}  & \textbf{3} &  &  &  & &  & & & & + & + & + \\   
    \hline
    \textbf{$A \rightarrow D$}  & \textbf{$e_8$}  & \textbf{2} & + &  &  & &  & & + & & & & \\   
    \hline
    \textbf{$B \rightarrow A$}  & \textbf{$e_9$}  & \textbf{2} & &  &  & & + & - & & & & & \\   
    \hline
    Not Selected & $e_{10}$  & 2 & & - &  & &  & & & - & & & \\   
    \hline
    Not Selected & $e_{14}$  & 2 & &  &  & &  & & - & - & & & \\   
    \hline
    Not Selected & $e_{23}$  & 2 &  &  &  & &  & & & & - & & -\\   
    \hline
    Not Selected & $e_1$  & 2 & &  &  & &  & & - & - & & & \\   
    \hline
    Not Selected & $e_2$  & 2 & &  &  & &  & & - & - & & & \\   
    \hline
    \textbf{$ABCD \rightarrow A$}  & \textbf{$e_4$}  & \textbf{2} &  &  & + & - &  & & & & & & \\   
    \hline
    \end{tabular}
    \end{small} 
\end{table*}}

{\renewcommand{\arraystretch}{2}%
\begin{table*}[!h]
\centering
    \caption{Timings for Implementation Strategy-I} \begin{small}
    \begin{tabular}{|c|c|c|c|c|c|c|c|c|c|c|}
    \hline
     {\bfseries Criteria } & {\bfseries $Exp_1$} & {\bfseries $Exp_2$} & {\bfseries $Exp_3$} 
     & {\bfseries $Exp_4$ } & {\bfseries $Exp_5$} & {\bfseries $Exp_6$} & {\bfseries $Exp_7$} 
     & {\bfseries $Exp_8$ } & {\bfseries $Exp_9$} & {\bfseries $Exp_{10}$}  
     \\
     
    \hline
    \multirow{1}{*}{\parbox{4cm}{\centering \textbf{\#edges in FDG}}}  & 1000 & 1000  & 1000  & 1000 & 1000 & 1000 & 1000 & 1000 & 1000 & 1000 \\
    \hline
    \multirow{1}{*}{\parbox{4cm}{\centering \textbf{\#edges in Forbidden Join Chains}}}  & 10 & 10 & 10 & 10 & 10 & 10 & 20 & 30 & 40 & 50  \\   
    \hline
    \multirow{1}{*}{\parbox{4cm}{\centering \textbf{\#Forbidden Join Chains}}}   & 20 & 40 & 60 & 80 & 100 & 100 & 100 & 100 & 100 & 100 \\   
    \hline
    \multirow{1}{*}{\parbox{4cm}{\centering \textbf{\#Required Attribute Sets}}}   & 50 & 50 & 50 & 50 & 50 & 50 & 50 & 50 & 50 & 50 \\   
    \hline
    \multirow{1}{*}{\parbox{4cm}{\centering \textbf{\#Join Chains per Required Attribute Set} }}  & 10 & 10 & 10 & 10 & 10 & 10 & 10 & 10 & 10 & 10 \\   
    \hline
    \multirow{1}{*}{\parbox{4cm}{\centering \textbf{\#edge per Join Chain for Required Attribute Set}}} & 10 & 10 & 10 & 10 & 10 & 10 & 10 & 10 & 10 & 10 \\   
    \hline
    \multirow{1}{*}{\parbox{4cm}{\centering \textbf{\#Duration}}}  & 2 ms & 2 ms & 3 ms & 3 ms & 4 ms & 3 ms & 3 ms & 4 ms & 4 ms & 6 ms \\   
    \hline
    \end{tabular}
    \end{small} 
\end{table*}}

{\renewcommand{\arraystretch}{2}%
\begin{table*}[!h]
\centering
    \caption{Timings for Implementation Strategy-II} \begin{small}
    \begin{tabular}{|c|c|c|c|c|c|c|c|c|c|c|}
    \hline
     {\bfseries Criteria } & {\bfseries $Exp_{11}$} & {\bfseries $Exp_{12}$} & {\bfseries $Exp_{13}$} 
     & {\bfseries $Exp_{14}$ } & {\bfseries $Exp_{15}$} & {\bfseries $Exp_{16}$} & {\bfseries $Exp_{17}$} 
     & {\bfseries $Exp_{18}$ } & {\bfseries $Exp_{19}$} & {\bfseries $Exp_{20}$}  
     \\
     
    \hline
    \multirow{1}{*}{\parbox{4cm}{\centering \textbf{\#edges in FDG}}}  & 1000 & 1000  & 1000  & 1000 & 1000 & 1000 & 1000 & 1000 & 1000 & 1000 \\
    \hline
    \multirow{1}{*}{\parbox{4cm}{\centering \textbf{\#edges in Forbidden Join Chains}}}  & 10 & 10 & 10 & 10 & 10 & 10 & 10 & 10 & 10 & 10  \\   
    \hline
    \multirow{1}{*}{\parbox{4cm}{\centering \textbf{\#Forbidden Join Chains}}}   & 100 & 80 & 60 & 40 & 20 & 100 & 100 & 100 & 100 & 100 \\   
    \hline
    \multirow{1}{*}{\parbox{4cm}{\centering \textbf{\#Required Attribute Sets}}}   & 50 & 50 & 50 & 50 & 50 & 50 & 50 & 50 & 50 & 50 \\   
    \hline
    \multirow{1}{*}{\parbox{4cm}{\centering \textbf{\#Join Chains per Required Attribute Set} }}  & 10 & 10 & 10 & 10 & 10 & 10 & 10 & 10 & 10 & 10 \\   
    \hline
    \multirow{1}{*}{\parbox{4cm}{\centering \textbf{\#edge per Join Chain for Required Attribute Set}}} & 10 & 10 & 10 & 10 & 10 & 10 & 20 & 30 & 40 & 50 \\   
    \hline
    \multirow{1}{*}{\parbox{4cm}{\centering \textbf{\#Duration}}}  & 2 ms & 2 ms & 1 ms & 1 ms & 1 ms & 1 ms & 1 ms & 1 ms & 1 ms & 1 ms \\   
    \hline
    \end{tabular}
    \end{small} 
\end{table*}}

\section{Experiments}

Constructing functional dependency graph of a given schema is a straightforward task and its time complexity is $O( \left| F^+\right| )$ which is negligable for a proactive process. Additionally, the time complexity of generating join chains  for given \quotes{Required} or \quotes{Forbidden} sets mainly depends on the, \quotes{all simple paths} \cite{son} algorithm, which corresponds to the main step of the process. Hence, also the overall algorithm is related to time complexity of DFS which is also negligible for a proactive solution. The main time consuming part of the whole process is the Algorithm-4, since it is an exhaustive algorithm. The algorithm can be implemented in two different directions:

\begin{itemize}

\item Checking each \quotes{Forbidden} set according to a brute-force selection on \quotes{Required} set (Implementation Strategy-I as given in Algorithm-4).
\item Checking each \quotes{Required} set according to a brute-force selection on \quotes{Forbidden} set (Implementation Strategy-II).
\end{itemize}

Since the problem is NP-complete some new heuristics may also be added to the algorithm which might reduce only the expected execution time. The experiments are performed by choosing meaningfull database parameters. It is important to note that, any logical database can be divided into sub-schemas in terms of join chains, so the algorithm can be repeated at each sub-schema easily. The timings are collected using a i7, 16 GB RAM machine and heuristics are used during implementation for a better result, such as checking the forbidden against allowed, or vice versa up to their counts.

Table-2 presents the results of the algorithm when each \quotes{Required} set is checked according to a brute-force selection on \quotes{Forbidden} set and respectively Table-3 depicts the results for the  implementation strategy-II.

These benchmarks show that the algorithm is scalable for even large database sizes. The algorithm is exponential for the worst case for both strategies, but, the better strategies can be desing according to the brute-force selection set size. As shown in the experiments, the timings are mainly dependedent on the brute-force selected set size.

\section{Related Work}

Database privacy and inference problem has been very popular and several works in this field influenced us in developing the new strategy proposed in this paper. Inference problem has been discussed in many papers \cite{r4,r5,r8,r9,r29}, but most of them are about reactive solutions. These kind of approaches include query rewriting mechanisms \cite{r7}, data perturbation methods \cite{r10,r16}, deception strategies\cite{x1} and decomposition-based approaches\cite{r12,A1,A2}. The basis of main data perturbation methods is “Differential Privacy”, whose idea is try to add  noise to the query result in order to prevent the identification while producing meaningful result. Differential Privacy method is a big step in the literature to prevent inference attacks; however its usage is limited to the statistical analysis. Another approach, known as deception mechanism, aims to corrupt the data by inserting anonymous data or structures, however its applicability is also limited as Differential Privacy.  K-anonymity \cite{r3} is a another leading research on this field, but differential privacy has proven the hardness of satisfying “non-identifiability” problem\cite{x2} for dynamic data distribution.

Application usage of database basically needs single-row identification with actual values (as described in the Introduction section). For this kind of applications, the methods developed to prevent the identification can be categorized as being reactive or proactive. Reactive methods tend to behave dynamically according to the policy or data distribution. Query rewriting techniques (as in Truman and Non-Truman Models called by Elisa Bertino et.al.'s paper \cite{r8}) are reactive solutions to the inference problem. Query history tracking mechanisms and Chinese-Wall method \cite{r11} like approaches are subject to performance issues, because of being reactive.

To determine the purpose \cite{r2} of the user during privacy protection is a major step. During these checks, attribute-based granularity \cite{r22} should be used to preserve precise privacy. The idea proposed in this paper is totally proactive, as normalization process, and can also be supported by reactive methods to construct a complete mechanism. Database security policy should be checked against visibility \cite{r6} requirements and the external layer should be constructed accordingly. This paper states a complete, optimized and applicable decomposition strategy compared to \cite{r12}, \cite{A1} and \cite{A2}. The aim of this paper is to perform the decomposition with policy check, minimal loss of dependencies and by taking care of indirect dependencies (called as ‘probabilistic dependency in \cite{r12}). The related works propose an effective way of decomposition database in a somehow similar manner; nevertheless this paper combines maximal availability, intended privacy, policy check and indirect dependencies to carry out a definite decomposition for the external layer.

\section{Conclusion and Future Work}

The approach given in this paper is a proactive cross-control of the required and the forbidden attribute sets of relational schema, which is achieved using a secure decomposition technique to produce an external schema with maximal availability and minimal loss of the dependencies. As a future work, this optimization process can be further improved by considering the query statistics of the users, and integrating applicable reactive control mechanisms. Even though the method presented in this paper is proactive, experiments show that the timings are acceptable even for a reactive-like behavior in future. We have experienced the benefits of this approach as given in a real-life example. 


\balance

\bibliographystyle{abbrv}
\bibliography{vldb_sample}

\end{document}